\documentclass[aps
,pra
,twocolumn
,letterpaper
,10pt
,tightenlines
,superscriptaddress
,nofootinbib
,tightenlines
,longbibliography
,notitlepage,floatfix]{revtex4-2}

\usepackage[english]{babel} 

\usepackage{tabulary}

\usepackage{datetime}

\usepackage{mathtools,amssymb,amsthm,nccmath,dsfont,physics,bm,stmaryrd} 

\SetSymbolFont{stmry}{bold}{U}{stmry}{m}{n}

\usepackage{xcolor,pict2e,tikz} 

\usepackage{verbatim}

\usepackage{setspace} 

\usepackage[inline]{enumitem}

\usepackage{relsize} 

\usepackage{adjustbox}

\usepackage[title]{appendix}

\usepackage{lipsum} 

\usepackage[OT2,T1]{fontenc}
\DeclareSymbolFont{cyrletters}{OT2}{wncyr}{m}{n}
\DeclareMathSymbol{\Sha}{\mathalpha}{cyrletters}{"58}

\usepackage{hyperref}
\hypersetup{
    colorlinks,
    linkcolor={blue!100!black},
    citecolor={blue!100!black},
    urlcolor={blue!100!black},
}

\usepackage{cleveref}

\theoremstyle{plain}
\newtheorem{theorem}{Theorem}             

\newtheorem{corollary}{Corollary}

\theoremstyle{definition}
\newtheorem*{definition*}{Definition} 

\theoremstyle{remark}

\crefname{theorem}{theorem}{theorems}
\Crefname{theorem}{Theorem}{Theorems}

\crefname{lemma}{lemma}{lemmas}
\Crefname{lemma}{Lemma}{Lemmas}

\crefname{proposition}{proposition}{propositions}
\Crefname{proposition}{Proposition}{Propositions}

\crefname{corollary}{corollary}{corollaries}
\Crefname{corollary}{Corollary}{Corollaries}

\crefname{definition}{definition}{definitions}
\Crefname{definition}{Definition}{Definitions}

\crefname{remark}{remark}{remarks}
\Crefname{remark}{Remark}{Remarks}

\let\originalleft\left
\let\originalright\right
\renewcommand{\left}{\mathopen{}\mathclose\bgroup\originalleft}
\renewcommand{\right}{\aftergroup\egroup\originalright}

\renewcommand{\Re}{\operatorname{Re}}

\makeatletter

\newcommand{\ringplus}{\mathbin{\text{\@ringplus}}}

\newcommand{\@ringplus}{%
  \ooalign{\hidewidth\raise1.3ex\hbox{\tiny$\circ$}\hidewidth\cr$\m@th+$\cr}%
}

\newcommand{\ringminus}{\mathbin{\text{\@ringminus}}}

\newcommand{\@ringminus}{%
  \ooalign{\hidewidth\raise0.9ex\hbox{\tiny$\circ$}\hidewidth\cr$\m@th-$\cr}%
}
\makeatother

\newcommand{\bvec}[1]{\bm{#1}}

\newcommand{\tp}[0]{\mathrm{T}}

\DeclareFontFamily{U}{wncy}{}
\DeclareFontShape{U}{wncy}{m}{n}{<->wncyr10}{}
\DeclareSymbolFont{mcy}{U}{wncy}{m}{n}
\DeclareMathSymbol{\Sh}{\mathord}{mcy}{"58}

\newcommand{\negspace}{\!}
\newcommand{\lsub}[2]{{\protect\vphantom{#1}}_{#2} \negspace {#1}}
\newcommand{\rsub}[2]{{#1} \negspace {\protect\vphantom{#1}}_{#2}}

\newcommand{\ketsub}[2]{\rsub {\ket{#1}} {#2}}
\newcommand{\brasub}[2]{\lsub {\bra{#1}} {#2}}

\newcommand{\qket}[1]{\ketsub{#1} q}

\newcommand{\outprodsubsub}[4]{\ketsub {#1}{#3} \brasub{#2}{#4}}

\newcommand{\qoutprod}[2]{\outprodsubsub{#1}{#2}q q}

\newcommand{\abss}[1]{\lvert{#1}\rvert}

\newcommand{\complex}[0]{\mathbb{C}}

\newcommand{\integers}[0]{\mathbb{Z}}

\renewcommand{\vec}[1]{\bm{\mathrm{#1}}}
\newcommand{\grn}{\color{green!60!black}}

\newcommand{\blk}{\color{black}}

\usepackage{graphicx}
\graphicspath{{./figures/}}

\let\vec\bvec

\usepackage{soul}

\newcommand{\symplectic}{\text{s}}

\newcommand{\sft}[1]{{\check{#1}}}

\renewcommand{\op}[1]{\smash[t]{\hat{#1}}}
\DeclareMathOperator{\Op}{Op}
\DeclareMathOperator{\Wig}{Wig}
\DeclareMathOperator{\Proj}{{\mathsf P}}
\DeclareMathOperator{\EforOps}{{\mathcal E}}
\DeclareMathOperator{\EforWigs}{{\mathsf E}}
\DeclareMathOperator{\Fourier}{{\mathsf F_{\symplectic;2d}}}

\newcommand{\inputdot}{\makebox[1.5ex]{\text{$\cdot$}}}
\newcommand{\CGS}{\textsc{cgs}}
\newcommand{\RS}{\textsc{rs}}
\newcommand{\DKS}{\textsc{dks}}

\newcommand{\twod}{{(2d)}}

\usepackage[makeroom]{cancel}

\renewcommand{\grn}{\blk}

\begin{document}

\allowdisplaybreaks

\title{Grand Unification of All Discrete Wigner Functions on $d \times d$ Phase Space}

\date{\today}

\author{Lucky K. Antonopoulos}
\email{Lucky.K.Antonopoulos@gmail.com}
\affiliation{Centre for Quantum Computation and Communication Technology, School of Science, RMIT University, Melbourne, Victoria 3000, Australia}

\author{Dominic G. Lewis}
\affiliation{Centre for Quantum Computation and Communication Technology, School of Science, RMIT University, Melbourne, Victoria 3000, Australia}

\author{Jack Davis}
\affiliation{DIENS, \'Ecole Normale Sup\'erieure, PSL University, CNRS, INRIA, 45 rue d'Ulm, Paris 75005, France}

\author{Nicholas Funai}
\affiliation{Centre for Quantum Computation and Communication Technology, School of Science, RMIT University, Melbourne, Victoria 3000, Australia}

\author{Nicolas C. Menicucci}
\email{nicolas.menicucci@rmit.edu.au}
\affiliation{Centre for Quantum Computation and Communication Technology, School of Science, RMIT University, Melbourne, Victoria 3000, Australia}

{
  \hypersetup{linkcolor=black}
}

\begin{abstract}
Wigner functions help visualise quantum states and dynamics while supporting quantitative analysis in quantum information. In the discrete setting, many inequivalent constructions coexist for each Hilbert-space dimension. This fragmentation obscures which features are fundamental and which are artefacts of representation. We introduce a stencil-based framework that exhausts all possible $d\times d$ discrete Wigner functions for a single $d$-dimensional quantum system (including a novel one for even $d$), subsuming known forms. We also give explicit invertible linear maps between definitions within the same $d$, enabling direct comparison of operational properties and exposing representation dependence.
\end{abstract}

\maketitle

\section{Introduction}
Quasiprobability representations provide a powerful framework for describing quantum systems by characterising states, operators, and measurements within a phase space. Spearheaded by Wigner's original function~\cite{wigner_Quantum_1932} for continuous systems, numerous discrete analogues for finite-dimensional Hilbert spaces exist~\cite{miquel_Quantum_2002, chaturvedi_Wigner_2010, Chaturvedi_Wigner_Weyl_2006, Ferrie_frames_2009, tilma_Wigner_2016, luis_Discrete_1998, leonhardt_Discrete_1996, horibe_Existence_2002, takami_Wigner_2001, gibbons_Discrete_2004, agam_Semiclassical_1995, saraceno_Translations_2019, bianucci_Discrete_2002, zak_Doubling_2011, raussendorf_Role_2023, gross_Hudsons_2006, royer_Wigner_1977, vourdas_Quantum_2004, lalita2023harnessing, abgaryan2021families,rundle_Overview_2021}, with many depending on number-theoretic properties of the dimension (e.g., prime, odd, even).

With so many dimension-specific definitions coexisting, we lack tools for comparing them in terms of what properties they share, how they differ (e.g.,~representational dependencies), and how operational features such as negativity---a recognised marker of contextuality and computational \emph{magic}~\cite{Heinrich_robustness_2019,Galvao_discrete_2005,mari_Positive_2012, Howard_Wallman_Veitch_Emerson_2014, Pashayan_2015, raussendorf_Role_2023,hahn_bridging_2025,davis_Identifying_2024,lalita2023harnessing,Cusumano_Non_stabilizerness_2025,Bjork_Klimov_Sanchez_Soto_2008}---translate from one definition to another.  
What would be desirable is to unify this fragmented landscape into a single cohesive framework, with a construction that exhausts all possible discrete Wigner functions (on a suitable phase space) and with tools to systematically study and compare their physical properties. In this work, we achieve this goal.

We introduce a unifying framework based on a \emph{Stencil Theorem} (Theorem~\ref{thm:stencil_theorem}), showing that every possible discrete Wigner function (DWF) for a $d$-dimensional qudit, defined on a $d\times d$ phase space, arises from cross-correlating a $2d\times 2d$ parent function with a $d$-dependent function called a \emph{stencil}. This exhausts all possible $d\times d$ DWF constructions and provides analytical tools to compare their physical predictions. We offer several examples of valid stencils, including one that generates a novel $d\times d$ DWF for all even dimensions.

\section{Notation}%
We work with two-component vectors of the form~$\bvec v = (v_1, v_2)^\tp$, whose symbols fix their domains:
    ${\bvec n
    \in \integers_N^2}$, $
    \bvec{m}
    \in \integers_{2d}^2$, $ 
    \bvec{\alpha}
    \in \integers_{d}^2$, $
    \bvec{b} 
    \in \integers_{2}^2$, $
    \bvec{k} 
    \in \integers^2$. Arithmetic for the first four is modulo~$N$, $2d$, $d$, and~$2$, respectively.
We define the $T$-periodic Kronecker delta, $\Delta_T[k]$, as being $1$ if ${k \bmod T = 0}$ and $0$ otherwise, extending naturally to vector arguments: $\Delta_T[\bvec v] = \Delta_T[v_1] \Delta_T[v_2]$.

\section{Discrete Wigner functions}
We now introduce the class of functions we consider and the domain on which they are defined.

\begin{definition*}[discrete phase space]
    A \emph{discrete phase space} is given by ${\mathcal{P}_N = \integers_N^2}$ for any $N \in \integers^+$.
\end{definition*}

Often, $N=d$ or $N=2d$ is used to represent a qudit of dimension~$d$, but we do not require that for the definition of the phase space itself. However, the following definition requires $N=d$.

\begin{definition*}[faithful phase-space representation]
    Given a discrete phase space~$\mathcal{P}_d$ and a linear, discrete-variable operator~$\op O \in \mathcal L(\complex^d)$, a \emph{faithful phase-space representation of~$\op O$}, denoted~$f_{\hat O}: \mathcal{P}_d \to \complex$, is a complex function on~$\mathcal{P}_d$ that is linear in $\op O$ and for which there exists a quantisation map (discrete analogue of the Weyl transform~\cite{weyl_quantenmechanik_1927}) $\Op$ such that $\Op [f_{\hat O}] = \op O$ for all $\op O$.
\end{definition*}

Here, $f_{\hat{O}}$ and $\Op$ take the form
\begin{align}
f_{\hat O}(\bm\alpha) &= \frac 1 d \Tr[\op{A}(\bm \alpha)^\dag\op O], &\Op[f] = \sum_{\bm\alpha} f(\bm\alpha)\op{B}(\bm\alpha),
\label{eq:weyl_rule}
\end{align}
where $\{\op{A}(\bm\alpha)\}$ and $\{\op{B}(\bm\alpha)\}$ are \emph{operator bases},%
\footnote{These are bases (and not just frames) because the operator space is finite dimensional.} %
i.e., spanning sets of the $d^2$-dimensional operator space $\mathcal{L}(\mathbb{C}^{d})$. Furthermore, $\Tr[\op{A}(\bm \alpha)^\dag\op B(\bvec \beta)] = d\, \Delta_d [\bvec \alpha - \bvec \beta]$, making these bases a pair of dual frames~\cite{Ferrie_frames_2008, Ferrie_frames_2009}.%
\footnote{We include the constant~$d$ in Eq.~\eqref{eq:weyl_rule}, which is typically omitted when defining dual frames.} %
In the continuous-variable setting, many functions arise from such pairs, including the Glauber-Sudarshan $P$-function~\cite{Glauber_function_1963, Sudarshan_function_1963, drummond_Generalised_1980}
, the Husimi $Q$-function~\cite{husimi_Formal_1940,appleby_Generalized_Husimi_1999}, and the Kirkwood-Dirac distribution~\cite{Kirkwood_Quantum_1933,Arvidsson-Shukur_Properties_2024}. 
The Wigner function~\cite{wigner_Quantum_1932}, however, arises from a \emph{self-dual frame}, which, for discrete systems, means that ${\op{A}(\bm\alpha) = \op{B}(\bm\alpha)}$ for all ${\bm\alpha \in \mathcal P_d}$.

Next, we introduce the Weyl-Heisenberg displacement operators \text{(WHDOs}),%
\footnote{We recommend the pronunciation ``who-doo.''} %
which correspond to translations in discrete phase space.

\begin{definition*}[WHDO]
    The Schwinger WHDOs%
    \footnote{\label{foot:WHDO_phase}
    Different WHDOs exist~\cite{durt_Mutually_2010,appleby_Generalized_Husimi_1999,bengtsson_Discrete_2017,raussendorf_Role_2023,schwinger_Unitary_1960}, dependent on choice of phase. Our choice, $\omega_{2d}^{-k_1 k_2}$, matches that used by Schwinger~\cite{schwinger_Unitary_1960} and will be justified below.}%
    ~\cite{schwinger_Unitary_1960} are 
    \begin{align}
    \op{V}(\bvec{k}) \coloneqq \omega_{2d}^{-k_1 k_2} \op{Z}^{k_2} \op{X}^{k_1},
    \label{WHDO_def}
\end{align}
where ${\omega_{d}^{a}\coloneqq e^{\frac{2\pi i}{d}a}}$ are the $d$th roots of unity, ${\op{Z}^{k_2}\coloneqq\sum_{j\in\integers_d} \omega_{d}^{k_2 j}\qoutprod{j}{j}}$ and ${\op{X}^{k_1}\coloneqq\sum_{j\in\integers_d} \qoutprod{j+k_1}{j}}$ are the clock and shift operators that generate displacements in discrete momentum and position, respectively, and $\qket{\inputdot}$ denotes the discrete position basis (a.k.a.\ computational basis)\blk, in which arithmetic is modulo~$d$.
\end{definition*}

\begin{definition*}[discrete PPO]
    For a $d$-dimensional qudit, a \emph{discrete phase-point operator (PPO)} 
    is an operator $\op A (\bvec n) \in \mathcal L(\complex^d)$ associated with a particular point~$\bvec n \in \mathcal P_N$ in a discrete phase space.%
    \footnote{Note that $N$ need not equal~$d$ for this definition.}
\end{definition*}

\begin{definition*}[valid PPO frame]
    A \emph{valid PPO frame} for a DWF is a set of $d^2$ discrete PPOs with ${N=d}$, i.e.,~${\{\op A (\bvec{\alpha}) \mid \bvec \alpha \in \mathcal P_d\}}$, that constitutes a Hermitian, trace-1, WHDO-covariant, self-dual frame. Explicitly, for all $\bvec \alpha, \bvec \beta \in \mathcal P_d$, the following hold:
    \begin{itemize}
    \item[A1:] (Hermiticity) $\op{A}(\bvec \alpha) = \op{A}(\bvec \alpha)^\dag$. \label{A1}
    
    \item[A2:] (Normalisation) ${\Tr[\op{A}(\bvec{\alpha})] = 1}$. \label{A2}
    
    \item[A3:] (Orthogonality) ${\Tr[\op{A}(\bvec{\alpha})^\dag \op{A}(\bvec{\beta})] = d~ \Delta_{d}[\bvec{\alpha} - \bvec{\beta}]}$.
    \label{A3}
    
    \item[A4:] (WHDO covariance) ${\op V(\bm k) \op A(\bvec{\alpha}) \op V^\dagger(\bm k) = \op A(\bvec{\alpha} + \bvec{k})}$.%
    \label{A4}
\end{itemize}
\end{definition*}
The criteria above can be seen as a discrete analogue~\cite{varilly_moyal_1989,Cahen_Stratonovich_2011,ferrie_Quasiprobability_2011, Heiss_Discrete_2000, raussendorf_Role_2023} of the Stratonovich-Weyl criteria~\cite{stratonovich_Distributions_1957,brif_General_1998} that underpin the continuous Wigner function. We can extend this concept of validity to DWFs themselves:

\begin{definition*}[valid DWF]
A \emph{valid DWF} representing an operator $\op{O} \in \mathcal L(\complex^d)$\grn, denoted~$W_{\hat O}:\mathcal P_d \to \complex$, is a faithful phase-space representation of~$\op O$ obtained via Eq.~\eqref{eq:weyl_rule} using a valid PPO frame.
\end{definition*}

We define one additional property not strictly required for validity but desirable nonetheless~\cite{gibbons_Discrete_2004,gross_Hudsons_2006, horibe_Existence_2002,leonhardt_Discrete_1996}.

\begin{definition*}[PPO marginals,  marginalisation]
    Given a valid PPO frame~$\{\op A(\bvec \alpha)\}$, its \emph{PPO marginals} are elements of the sets  $\{\op Q_H(\alpha_2) \coloneqq \frac 1 d \sum_{\alpha_1} \op A(\bvec \alpha)\}$ and $\{\op Q_V(\alpha_1) \coloneqq \frac 1 d \sum_{\alpha_2} \op A(\bvec \alpha)\}$---i.e.,~averages of the PPOs along horizontal and vertical lines, respectively.%
    \footnote{While some authors consider more general lines~\cite{leonhardt_Discrete_1996,miquel_Quantum_2002,horibe_Existence_2002,wootters_Wignerfunction_1987}, doing so restricts the qudit dimension \emph{a priori}, so we consider rectilinear lines only, leaving generalisations to future work.} %
    A valid PPO frame satisfies \emph{marginalisation} if its PPO marginals form a mutually unbiased bases~\cite{Srinivasan_Generalized_2018,ivonovic1981geometrical,durt_Mutually_2010} for the qudit Hilbert space, $\complex^d$.
\end{definition*}

\section{Doubled DWF}%
Other quasiprobability representations exist that resemble a valid DWF but fail to satisfy all defining criteria (A1--A4). One notable example, appearing often in the literature~\cite{hannay_Quantization_1980,leonhardt_Discrete_1996,miquel_Quantum_2002,arguelles_Wigner_2005,Morgan_Tracy_thesis,zak_Doubling_2011, 
hahn_Quantifying_2022,feng_Connecting_2024}, is what we call the \emph{doubled DWF}, defined on $\mathcal{P}_{2d}$.%
\footnote{\label{foot:P2d_redundancy}Since it is defined over $\mathcal P_{2d}$, the doubled DWF contains $4d^2$ elements and thus exhibits a fourfold redundancy in the encoded information, seen as repeated values across its four $d\times d$ subgrids, up to sign changes~\cite{leonhardt_Discrete_1996}.} %
It is the parent function for all valid DWFs and the starting point for our stencil-based framework, detailed below. We obtained this function, like Feng and Luo~\cite{feng_Connecting_2024}, from the Gottesman-Kitaev-Preskill~(GKP) encoding of a qudit into a continuous Hilbert space~\cite{gottesman_Encoding_2001} and considering one unit cell of the resulting continuous Wigner function. Choosing to use the GKP encoding also selects the phase for our WHDOs (see footnote~\ref{foot:WHDO_phase}).

The doubled DWF for an operator~$\op O$, denoted $W_{\hat{O}}^\twod$, is obtained using the doubled PPO frame ${\{\op A^\twod(\bvec m) \mid \bvec m \in \mathcal P_{2d}\}}$ and satisfies the following:
\begin{subequations}
\label{dbld_DWF_and_PPO}
\begin{align}
    &W^\twod_{\hat{O}}(\bvec{m}) 
    =
    \Wig^\twod[\op O](\bvec m)
    \coloneqq 
    \frac{1}{2d} \Tr[\op{A}^\twod(\bvec{m})^\dag\op{O} ], 
    \label{eq:dbld_Wig_op_expanded} \\
    &\op O = 
    \Op^\twod[ W^\twod_{\hat{O}}] \coloneqq 
    \frac 1 2 \sum_{\bvec m} W^\twod_{\hat{O}}(\bvec{m}) \op A^\twod(\bvec m),
    \label{eq:dbld_O_op_expanded} \\
    &\op{A}^\twod(\bvec{m})
    \coloneqq 
    \op{V}(\bvec{m})\op{R}
    = \op{R}\op{V}(\bvec{m})^\dag,
    \label{eq:dbld_PPO}
\end{align}%
\end{subequations}
where the superscript $\twod$ indicates that the phase space is doubled, $\op{O} \in \mathcal L(\complex^d)$ is the qudit operator being represented, and ${\op{R}\coloneqq\sum_{j\in\integers_d}\qoutprod{-j}{j}}$ is the discrete parity operator (using arithmetic modulo~$d$). Also, note that the factor of $\frac 1 2$ in Eq.~\eqref{eq:dbld_O_op_expanded} is a consequence of working with a doubled DWF that uses a normalisation of $\frac 1 {2d}$ instead of the usual $\frac 1 d$.  Additionally, we call the operations~$\Wig^\twod$ and $\Op^\twod$ the \emph{doubled Wigner transform} and \emph{doubled Weyl transform} (analogues of the CV versions~\cite{weyl_quantenmechanik_1927}), respectively. They will be useful later on.

\section{Cross-correlation with a stencil}
We show that \emph{every valid DWF} can be obtained by cross-correlating a single parent function---i.e.,~the doubled DWF---with a chosen function (stencil) on the doubled phase space. Moreover, we show that all quantities calculable from one valid DWF---e.g.,~negativity---must also appear in some (generally different) form in any other valid DWF for the same~$d$\blk.
We now introduce the stencils, together with the DWFs and PPOs they generate.

\begin{definition*}[stencil]
    A \emph{stencil}~$M$ is any complex function of the doubled phase space, $M : \mathcal P_{2d} \to \complex$.
\end{definition*}

\begin{definition*}[$M$-DWF, $M$-PPO]
    Let~$M$ be a stencil and let $\op{O} \in \mathcal L(\complex^d)$ be a qudit operator. We define the \emph{\text{$M$-DWF} generated by~$M$} (for the operator~$\op O$) and the \emph{\text{$M$-PPO} generated by $M$}, respectively, as
    \begin{subequations}
    \label{eqs:M_cross_cor_defs}
    \begin{align}
        W_{\hat{O}}^{M}(\bvec{\alpha}) &\coloneqq 
        (M \star W_{\hat{O}}^\twod)(2\bvec{\alpha})
        = \frac 1 d \Tr[\op A^M(\bvec \alpha)^\dag \op O],
        \label{eq:WM_def} \\
        \op{A}^{M}(\bvec{\alpha}) &\coloneqq \frac{1}{2} (M^{*} \star \op{A}^\twod )(2\bvec{\alpha}),
        \label{eq:AM_def}
    \end{align}
    \end{subequations}
    where $\star$ indicates cross-correlation [see also Eq.~\eqref{eq:cross_cor_as_inner_prod}], ${(f \star g)(\bvec m)} = \sum_{\bvec m'} f(\bvec m')^* g(\bvec m' + \bvec m)$. 
\end{definition*}

The cross-correlations in Eqs.~\eqref{eqs:M_cross_cor_defs} use the stencil~$M$ to reorganise relevant information onto even grid sites and then discard the odd ones%
\footnote{Even-integer sites are chosen because odd coordinates (half-integers in~\cite{leonhardt_Discrete_1996}) are typically introduced post hoc to accommodate both even-$d$ and odd-$d$ systems. Additionally, the marginals of the doubled DWF over odd-integer sites are always zero~\cite{leonhardt_Discrete_1996}.} %
(by evaluating at~${\bvec m = 2\bvec \alpha}$), mapping from $\mathcal P_{2d}$ to $\mathcal P_d$, with the choice of stencil determining how this is done.

\section{Stencil criteria}%
As presented above, the DWF involves a cross-correlation of the doubled DWF with a stencil $M$. Consequently, for the resulting $M$-DWF to satisfy A1--A4, $M$ must also satisfy some analogous criteria. Since A4 is always satisfied for all $M$,%
\footnote{\label{foot:A2d_WHDO}
The doubled PPO satisfies $\op V(\bvec k) \op A^\twod(\bvec m) \op V(\bvec k)^\dag = \op A^\twod(\bvec m + 2 \bvec k)$, so all $M$-PPOs automatically satisfy A4 since $\star$ is linear in the second argument and is being evaluated at~$2\bvec \alpha$.} %
we focus here on the remaining three.

We start by defining an orthogonal projection map on the doubled phase space,~$\Proj \coloneqq \Wig^\twod \circ \Op^\twod$ [see Eqs.~\eqref{dbld_DWF_and_PPO}], acting on a general function, $f:\mathcal P_{2d} \to \complex$, as
\begin{align}
\label{eq:Pdef}
    \bar f \coloneqq \Proj f 
    = 
    \Wig^\twod \circ \Op^\twod[f]
    =
    W^\twod_{\Op^\twod[f]}.
\end{align}
The explicit form of this projector is given in the Appendix. Note that we %
define the shorthand~$\bar f$ for the projected version of~$f$. We use this projector to state the criteria for~$M$ to generate a valid $d \times d$ DWF:

\begin{definition*}[Valid stencil and~$\mathcal M$]
A stencil~$M$ is called \emph{valid} if its projection $\bar M$ satisfies all of the following:%
\footnote{M1--M3 can also be expressed in the Fourier domain using the symplectic discrete Fourier transform---see the Appendix.\blk} %
\begin{enumerate}[align=left]
    \item[M1:] $\bar M(\bvec{m})^* =\bar M(\bvec{m})$, \label{M1}
    
    \item[M2:] 
    $\sum\limits_{\bvec{m}} \bar M(\bvec{m}) = 1$,
    \label{M2}
    
    \item[M3:] $(\bar M \star \bar M)(2 \bvec \alpha) = \Delta_d[\bvec \alpha].$
\end{enumerate}
We define $\mathcal M$ as the set of all valid stencils.
\end{definition*}

As our main result, we present the following theorem, with proof supplied in the Appendix:
\begin{theorem}[Stencil Theorem]
\label{thm:stencil_theorem}
    (a) Every valid DWF over a $d\times d$ phase space is an $M$-DWF generated by some valid stencil. (b) Every valid stencil $M \in \mathcal M$ generates a valid $M$-DWF with corresponding valid $M$-PPO frame~$\{\op A^M(\bvec \alpha)\}$.
\end{theorem} 

The Stencil Theorem shows (a)~valid stencil mappings of the doubled DWF (with stencils satisfying M1--M3) exhaust all possible valid $d\times d$ DWFs (those whose PPOs~$\op A(\bvec \alpha)$ satisfy A1--A4), and (b)~every valid stencil generates a valid $d\times d$ DWF. One generating stencil for the DWF in (a) is $M
= \Wig^\twod [\op A(\bvec 0)] = \bar M$ (see the Appendix).

Finally, note that for a given valid DWF, its generating stencil~$M$ is not unique, but its projected stencil~$\bar M$ is unique (for proof, see the Appendix):

\begin{corollary}\label{corr:nonunique_stencil}
Any two valid stencils~$M_1, M_2 \in \mathcal M$ generate the same $M$-DWF if and only if $\bar M_1 = \bar M_2$ .
\end{corollary}

\section{Examples of valid stencils}%
Here we introduce some examples of valid stencils.
The \emph{reduction stencil} (RS)~$M^\RS(\bvec m) \coloneqq 2 \Delta_{2d} [\bvec m]$, valid for odd~$d$, is motivated by Leonhardt's work~\cite{leonhardt_Discrete_1996}, which reveals that a single copy of non-redundant information already exists at the even grid sites, so the odd ones can simply be discarded. Its $M$-PPO frame, $\{\op A^\RS(\bvec \alpha) =  \op A^\twod(2 \bvec \alpha)\}$, satisfies marginalisation.

For even~$d$, the even grid sites in the doubled DWF have a fourfold redundancy (see footnote~\ref{foot:P2d_redundancy}), so we need a stencil that also involves the odd ones.%
\footnote{Like other authors~\cite{leonhardt_Discrete_1996,zak_Doubling_2011}, we conjecture that the dichotomy between even and odd~$d$ may be linked to spin statistics.} %
The \emph{coarse-grain stencil} (CGS)~${M^\CGS(\bvec m) \coloneqq \sum_{\bvec b} \Delta_{2d}[ \bvec m - \bvec b]}$, valid for even~$d$ and extrapolated from its known $d=2$ incarnation [Eq.~\eqref{PPOs_RS_Leon_Cohen_Woot1}, below], averages four neighbouring sites into one, ensuring we use the information from both even and odd sites. Its $M$-PPO frame, $\{\op A^\CGS(\bvec \alpha) = \frac 1 2 \sum_{\bvec b} \op A^\twod(2 \bvec \alpha + \bvec b)\}$, satisfies marginalisation.

Figure~\ref{fig:stencilplots} illustrates these two stencils---along with their respective projections, $\bar M$ [using Eq.~\eqref{eq:Pdef}]---for two representative dimensions. The Appendix proves both stencils' validity, as well as marginalisation for both $M$-PPOs.

By inspection [see Eq.~\eqref{eq:dbld_PPO}] and through Ref.~\cite{ferrie_Quasiprobability_2011}, we find that the PPO $\op A^\RS$ has the same form as Wootters'~\cite{wootters_Wignerfunction_1987} for prime $d > 2$, as well as Gross'~\cite{gross_Hudsons_2006}, Leonhardt's~\cite{leonhardt_Discrete_1996}, and (up to a sign flip on~$\alpha_1$) Cohendet et~al.'s Fano operator~$\op \Delta$~\cite{cohendet_Stochastic_1988,fano_Description_1957} for odd~$d$. Furthermore, $\op A^\CGS$ matches Wootters' definition when evaluated for $d=2$, which was also shown by Feng and Luo~\cite{feng_Connecting_2024}. But for even $d > 2$ the coarse-grain DWF lies outside the family of $d \times d$ DWFs defined by Chaturvedi et~al.~\cite{chaturvedi_Wigner_2010}, rendering it a novel $d \times d$ DWF for even~$d$. 
Altogether, we have the following equivalences (possibly among others):
\begin{align}
    \op{A}^{\RS}_{{\text{odd $d$}}}
    &=
    \op{A}^{\text{\cite{wootters_Wignerfunction_1987}}}_{{\text{prime}~d>2}}
    = 
    \op{A}^{\text{\cite{leonhardt_Discrete_1996}}}_{{\text{odd $d$}}}
    =
    \op{A}^{\text{\cite{gross_Hudsons_2006}}}_{{\text{odd $d$}}}
    =
    \op{\Delta}^{\text{\cite{cohendet_Stochastic_1988}}}_{{\text{odd $d$}}}(-\inputdot,\inputdot)
    ,
    \nonumber \\*
    \op{A}^{\CGS}_{{d=2}} &= \op{A}^{\text{\cite{wootters_Wignerfunction_1987}}}_{{d=2}}
    =
    \op{A}^{\text{\cite{chaturvedi_Wigner_2010}}}_{{d=2}},
    \label{PPOs_RS_Leon_Cohen_Woot1}
\end{align}

The Appendix introduces a third stencil based on a Dirichlet kernel~\cite{agam_Semiclassical_1995}, $M^\DKS$, which is valid for odd~$d$ but whose projection~$\bar M^\DKS$ differs from $\bar M^\RS$. Thus, by Corollary~\ref{corr:nonunique_stencil}, this generates a DWF distinct from Gross’~\cite{gross_Hudsons_2006}.

\begin{figure}[t]
\centering
\includegraphics[width=\columnwidth]{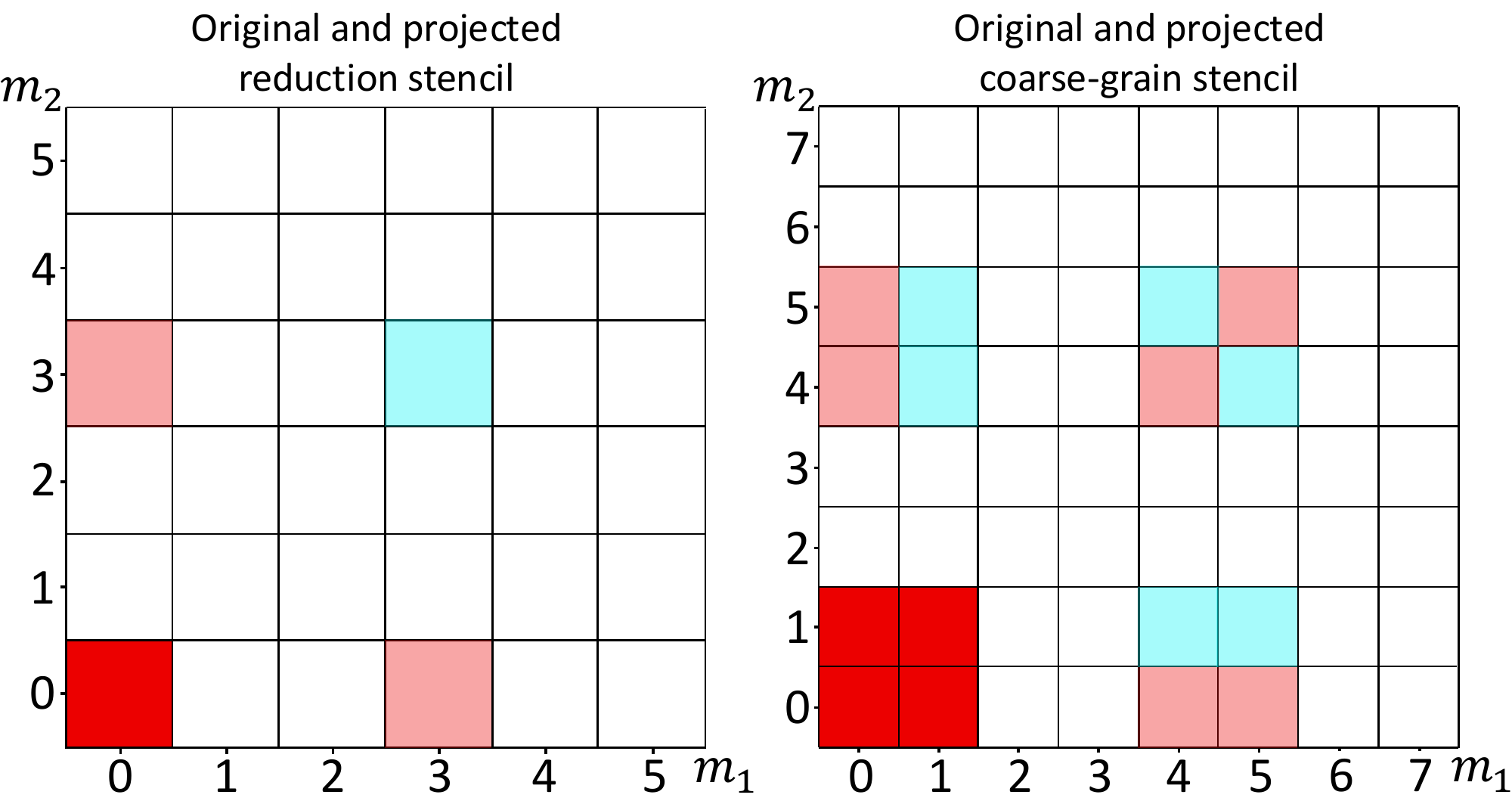}
\caption{Stencil plots of $M(\bvec{m})$ (bright) and its projection~$\bar M(\bvec{m})$ (bright+faded) for the reduction stencil $M^{\RS}$ (${d=3}$, left) and the coarse-grain stencil $M^{\CGS}$ (${d=4}$, right). White, red, and cyan (of any saturation) indicate $0$, $+c$, $-c$, respectively, with a different constant~$c$ for each stencil: $c^\RS = 2$ and $c^\CGS= 1$; for the projected stencils, $\bar c^\RS = \tfrac 1 2$ and $\bar c^\CGS = \tfrac 1 4$.}
\label{fig:stencilplots}
\end{figure}

\section{Linear maps corresponding to a change of stencil}
Since all valid DWFs are $M$-DWFs, we have a systematic way to link all of them (for the same qudit dimension) based solely on their defining stencils, which are guaranteed to exist by Theorem~\ref{thm:stencil_theorem}.

\begin{theorem}
    There exist (a)~a stencil-dependent linear map between any two operators reconstructed from the same phase-space function but using different valid PPO frames and (b)~a stencil-dependent linear map between any two valid DWFs representing the same operator.
    \label{thm:E_existence}
\end{theorem}
\begin{proof}
    Let $M_1, M_2 : \mathcal P_{2d} \to \complex$ be arbitrary valid stencils for the same doubled phase space. For each such stencil pair~$(M_1, M_2)$, define the stencil-dependent linear maps $\EforOps : \mathcal L(\complex^d) \to \mathcal L(\complex^d)$ and $\EforWigs: (\mathcal{P}_d \to \complex) \to (\mathcal{P}_d \to \complex)$ acting on operators and phase-space functions, respectively:
    \begin{align}
    \label{eq:Emap_ops}
        \EforOps_{M_{1 \to 2}} &\coloneqq \Op^{M_2} \circ \Wig^{M_1},
    \\
    \label{eq:Emap_functions}
    \EforWigs_{M_{1 \to 2}} &\coloneqq \Wig^{M_2} \circ \Op^{M_1},
    \end{align}
    where ${\Wig^M[\op O] \coloneqq W^M_{\hat O}}$, ${\Op^M[f] \coloneqq \sum_{\bvec \alpha} f(\bvec \alpha) \op A^M(\bvec \alpha)}$, and $M_{1 \to 2}$ is short for $M_1 \to M_2$. For any pair of valid PPO frames for the same qudit dimension~$d$, a corresponding valid stencil pair~$(M_1, M_2)$ always exists (Theorem~\ref{thm:stencil_theorem}). Thus, $\EforOps_{M_{1 \to 2}}$ and $\EforWigs_{M_{1 \to 2}}$ always exist and satisfy the following for all $\op O \in \mathcal L(\complex^d)$ and all phase-space functions~$W:\mathcal P_d \to \complex$:
    \begin{align}
         \EforOps_{M_{1\rightarrow2}}\{\Op^{M_1}[W]\}
         &=
         \Op^{M_2}[W],
        \\
        \EforWigs_{M_{1\rightarrow 2}}\{\Wig^{M_1}[\hat{O}]\} &= \Wig^{M_2}[\hat{O}].
    \end{align}
    This proves the Theorem.
\end{proof}
\noindent Figure~\ref{fig_Superoperator_Maps} shows a graphical representation of this Theorem. We discuss its implications below.

\begin{figure}[t]
\centering
\includegraphics[width=\columnwidth]{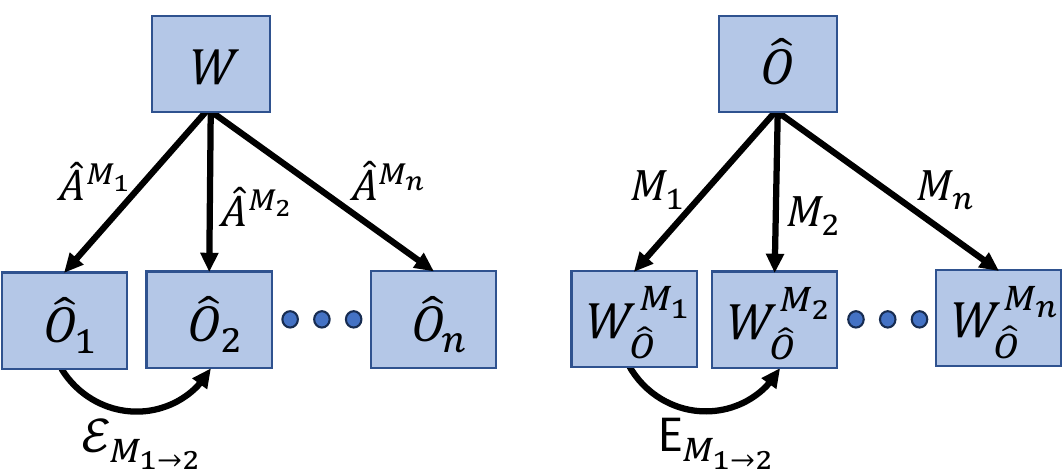}
\caption{Left: A function on phase space, $W:\mathcal P_d \to \complex$, can represent different operators $\op O_i$ by varying the choice of $M$-PPO~$\op A^{M_i}$ used to construct it. These are all related by a stencil-dependent linear map on operators, $\EforOps$. Right: Analogously, each valid DWF can represent a given operator~$\op O$ as a distinct phase-space function, all of which are related by a stencil-dependent linear map on phase-space functions,~$\EforWigs$. By Theorem~\ref{thm:E_existence}, these maps exist for all valid PPO frames.}
\label{fig_Superoperator_Maps}
\end{figure}

\section{Discussion}
Our stencil-based framework systematically constructs all valid $d\times d$ DWFs by cross-correlating the doubled DWF, whose entries contain a fourfold redundancy (see footnote~\ref{foot:P2d_redundancy}), with a stencil $M\in\mathcal M$ that removes this redundancy. Valid DWFs for the same~$d$ differ only in the choice of stencil---e.g., $M^\RS$ versus $M^\DKS$ in odd dimensions and $M^\CGS$ versus others for even~$d$~\cite{chaturvedi_Wigner_2010}.

The invertible linear maps ($\EforOps$, $\EforWigs$) unify all valid DWFs into a single equivalence class for each Hilbert-space dimension, making differences between DWFs (for fixed~$d$) purely representational. These maps, therefore, enable representation-independent benchmarking of quantum resource measures~\cite{Howard_Campbell_RoM_2017,Heinrich_robustness_2019,Leon_Stabilizer_2022,raussendorf_Role_2023,davis_Identifying_2024} and systematic comparisons of features depending on the chosen stencil---e.g.,\ $M^\RS$ versus $M^\DKS$ in odd dimensions and $M^\CGS$ versus others for even~$d$~\cite{chaturvedi_Wigner_2010}.

Essential topological differences between dimensions~\cite{raussendorf_Role_2023} leave the most useful features from the odd-$d$ setting unavailable for even~$d$---notably, Hudson's theorem~\cite{gross_Hudsons_2006}, Clifford covariance~\cite{Zhu_permutation2016,raussendorf_Role_2023}, and hence a straightforward resource theory of negativity~\cite{Galvao_discrete_2005,gross_Hudsons_2006,raussendorf_Role_2023}. Revisiting these features from a stencil perspective may clarify the nature of these limitations and inspire the discovery of novel, dimension-agnostic features and resources.

Recasting the Wigner framework in terms of stencils converts a historical obstacle---redundancy removal~\cite{zak_Doubling_2011,leonhardt_Discrete_1996,agam_Semiclassical_1995,arguelles_Wigner_2005,feng_Connecting_2024}---into a purposeful design choice. The stencils themselves naturally provide a systematic catalogue of all valid DWFs, offering an organised search space to identify optimal candidates for efficient classical simulation of quantum circuits via phase-space methods~\cite{Veitch_Ferrie_Gross_Emerson_2012,mari_Positive_2012,Pashayan_2015,Kocia2017,Park_Extending_2024}. Furthermore, by relaxing specific validity criteria (criterion A1, and consequently M1), our approach generalises directly to discrete Kirkwood–Dirac quasidistributions~\cite{Kirkwood_Quantum_1933,Arvidsson-Shukur_Properties_2024}, with potential extensions to other quasiprobability families. Finally, future work may explore the connection of our stencil framework to Cohen's class of quasidristributions~\cite{cohen_Generalized_1966}, extended to discrete signals~\cite{wu1994discrete}.

\emph{Acknowledgments---}We thank Takaya Matsuura, Ben Baragiola, Ryuji Takagi, and Julian Greentree for their support in this work. This work was supported by the Australian Research Council (ARC) through the Centre of Excellence for Quantum Computation and Communication Technology (Project No. CE170100012). N.C.M. was supported by an ARC Future Fellowship (Project No. FT230100571). J.D. acknowledges funding from the Agence Nationale de la Recherche (ANR) under the Plan France 2030 programme, project NISQ2LSQ (Grant No. ANR-22-PETQ-0006) and from the European Union’s Horizon Europe Framework Programme (EIC Pathfinder Challenge project Veriqub) under Grant Agreement No. 101114899.

\appendix
\section*{Appendix}
\setcounter{section}{1}
\setcounter{equation}{0}
\renewcommand{\theequation}{A\arabic{equation}}

\subsection{Stencil projector~$\Proj$ and inner products on $\mathcal P_{2d}$}%
The doubled phase space $\mathcal{P}_{2d}$ can be equipped with an inner product $\langle f, g \rangle \coloneqq \sum_{\bvec m} f(\bvec m)^* g(\bvec m)$ to form $L_{2}(\mathcal{P}_{2d})$. $\Proj$ is a linear map on this inner-product space. Thus,
\begin{align}
    \bar f(\bvec m) = (\Proj f)(\bvec m) &= \sum_{\bvec m'} P_{\bvec m, \bvec m'} f(\bvec m'),
    \label{eq:proj_of_func}
\end{align}
where, using Eqs.~\eqref{WHDO_def}, \eqref{dbld_DWF_and_PPO}, and~\eqref{eq:Pdef}, we get its matrix elements explicitly:
\begin{align}
    P_{\bvec m, \bvec m'} &\coloneqq \frac 1 {4d} \Tr[\op A^\twod(\bvec{m})^{\dag} \op A^\twod(\bvec{m}')]
\\
    &= \frac 1 4 \sum_{\bvec b} 
    (-1)^{b_1 m_2 -b_2 m_1 - b_1 b_2 d}\Delta_{2d}[\bvec m - \bvec m' - d\, \bvec b]. \nonumber
\end{align}
Altogether, then, for any function~$f:\mathcal P_{2d}\to \complex$,
\begin{align}
\label{eq:proj_explicit}
    \bar f(\bvec m)
&=
    \frac 1 4 \sum_{\bvec b} 
    (-1)^{b_1 m_2 -b_2 m_1 - b_1 b_2 d} f(\bvec m - d\, \bvec b).
\end{align}

Since $\op A^\twod$ is Hermitian, $P_{\vec m, \vec m'} = P_{\vec m', \vec m}^*$, ensuring $\Proj$ is self-adjoint with respect to $L_{2}(\mathcal{P}_{2d})$. 
From Eq.~\eqref{eq:Pdef}, we have that $\Proj$ is idempotent ($\Proj^{2}=\Proj$); hence it is an orthogonal projector on this space. This means
\begin{align}
\label{eq:dbl_inner_prod}
    \langle \bar f, g \rangle
    =
    \langle f, \bar g \rangle
    =
    \langle \bar f, \bar g \rangle
    =
    \frac 1 d \Tr[\Op^\twod[f]^\dag \Op^\twod[g]]
\end{align}
for all $f, g \in L_2(\mathcal P_{2d})$, which is an extension of the \emph{inner product property} for the doubled DWF~\cite{leonhardt_Discrete_1996} to include all functions~$f, g$ on the doubled phase space.

Furthermore, since $\sum_{\bvec m} \hat V(\bvec m) = 2d\, \op R$ by direct evaluation, $\sum_{\bvec m} \op A^{(2d)}(\bvec m) = 2d\, \hat{\mathbb{I}}$, where $\hat{\mathbb{I}}$ is the identity operator, using Eq.~\eqref{eq:dbld_PPO}. Now, let $1$ be a constant function equal to the value~1 on all of $\mathcal P_{2d}$. Then, $\Op^\twod[1] = d\, \hat{\mathbb{I}}$. This gives the \emph{trace property} for the doubled DWF~\cite{leonhardt_Discrete_1996}:
\begin{align}
\label{eq:dbl_trace_prop}
    \sum_{\bvec m} \bar g(\bvec m)
    =
    \langle 1, \bar g \rangle
    =
    \Tr [\Op^\twod[g]].
\end{align}
Finally, note that the cross-correlation is a shifted inner product:
\begin{align}
\label{eq:cross_cor_as_inner_prod}
    (f \star g)(\bvec m)
    =
    \langle f, g( \inputdot + \bvec m) \rangle
    =
    \langle f ( \inputdot - \bvec m), g \rangle
\end{align}
and is related to convolution~($*$) by reflecting and conjugating the first function: $f * g = f(- \inputdot)^* \star g$.

\subsection{Symplectic discrete Fourier transform~(SDFT)}
The SDFT of a function on the doubled phase space, ${\Fourier : (\mathcal P_{2d} \to \complex) \to (\mathcal P_{2d} \to \complex)}$, notated~$\sft f = \Fourier f$, is
\begin{align}
\label{eq:SDFT}
    \sft f(\bvec m)
&\coloneqq
    \frac 1 {2d}
    \sum_{\bvec m'}
    \omega_{2d}^{-(m_1 m_2' - m_2 m_1')}
    f(\bvec m'),
\end{align}
so named because $(m_1 m_2' - m_2 m_1')$ is the symplectic inner product of~$\bvec m$ and $\bvec m'$. Unlike the ordinary discrete Fourier transform, the SDFT~$\Fourier$ is self-inverse ($\Fourier^2 = \operatorname{id}$). Furthermore, although not obvious, $\Fourier \Proj = \Proj \Fourier$, and thus, $\sft{\bar f} = \bar{\sft f}$ for any function~$f$. (Express $\Fourier \Proj f$ as a double sum, exchange the order of summation, and resum to reveal $\Proj \Fourier f$.) This will be important in deriving the validity conditions on $\sft M$, discussed next.

\subsection{Stencil criteria in the Fourier domain}
The criteria M1--M3 for a stencil~$M$ to be valid become even simpler in terms of its SDFT~$\sft M$ (defined above). That is, a stencil~$M$ is valid if and only if the projection of its SDFT, $\bar{\sft M} = \Proj \Fourier M$, satisfies all three of the following:
\begin{enumerate}[align=left]
    \item[\qquad$\sft{\text M}1$:] $\bar{\sft M}(\bvec{m})^* = \bar{\sft M}(-\bvec{m})$, \label{sftM1}
    
    \item[\qquad$\sft{\text M}2$:] 
    $\, \bar{\sft M}(\bvec 0) = \dfrac 1 {2d}$,
    \label{sftM2}
    
    \item[\qquad$\sft{\text M}3$:] $\bigl\lvert\bar{\sft M}(\bvec m)\bigr\rvert = \dfrac 1 {2d}.$
    \label{sftM3}
\end{enumerate}
These are obtained by taking the SDFT of M1--M3, respectively, to obtain conditions on $\sft{\bar M} = \Fourier \Proj M$ and then noting (see above) that $\sft{\bar M} = \bar{\sft M}$. Criterion~$\sft{\text M}3$ requires two additional step: starting from M3, replace $2\bvec \alpha \mapsto \bvec m$ and then multiply both sides by~$\Delta_2[\bvec m]$ to ensure the condition only applies to even values of~$\bvec m$. This gives $(\bar M \star \bar M)(\bvec m) \Delta_2[\bvec m] = \Delta_{2d}[\bvec m]$ (since $\Delta_{2d}[\bvec m]\Delta_2[\bvec m]$ is just $\Delta_{2d}[\bvec m]$). Then, taking the SDFT of both sides gives $\sum_{\bvec b} \abss{\bar{\sft M}(\bvec m - d \bvec b)}^2=\frac 1 {d^2}$. Finally, using the recurrence relation $\bar{\check M}(\bvec m- d\bvec b) = (-1)^{m_1b_2-m_2b_1-b_1b_2 d} \bar{\check M}(\bvec m)$, obtained through Eq.~\eqref{eq:proj_explicit}, and taking a square root, we obtain $\sft{\text M}3$ as shown.

\subsection{Proof of Theorem~\ref{thm:stencil_theorem}}
We begin by proving part~(b). The cross-correlation defining an $M$-PPO [Eq.~\eqref{eq:AM_def}] can be written as $\op A^M(\bvec \alpha) = \Op^\twod[M(\inputdot - 2 \bvec \alpha)^*]$, using the right-most expression in Eq.~\eqref{eq:dbld_O_op_expanded}. From the first equality in Eq.~\eqref{eq:dbld_O_op_expanded}, $\op O = \Op^\twod[W^\twod_{\hat O}]$ for any~$\op O$, we see that 
\begin{align}
\label{eq:Op_Wig_id}
    \Op^\twod \circ \Wig^\twod = \operatorname{id},
\end{align}
the identity map on operators, and thus,
\begin{align}
    \label{eq:P_on_Wig}
    \Proj \circ \Wig^\twod &= \Wig^\twod, \\
    \label{eq:Op_on_P}
    \Op^\twod \circ \Proj &= \Op^\twod,
\end{align}
which means we can also write
\begin{align}
\label{eq:AM_from_Mbar}
    \op A^M(\bvec \alpha) 
    = 
    \Op^\twod[\bar M(\inputdot - 2 \bvec \alpha)^*]
\end{align}    
(note the bar over~$M$). With this expression in hand, M1 implies that $\op A^M(\bvec \alpha)$ is Hermitian~(A1). M2 and the $W^\twod$ trace property, Eq.~\eqref{eq:dbl_trace_prop}, together imply A2. Finally, M3, along with Eq.~\eqref{eq:cross_cor_as_inner_prod} and the $W^\twod$ inner product property, Eq.~\eqref{eq:dbl_inner_prod}, show that $(\bar M \star \bar M)(2\bvec \alpha) = \frac 1 d \Tr[\op A^M(\bvec 0)^\dag \op A^M(\bvec \alpha)]$. PPO covariance (A4; see footnote~\ref{foot:A2d_WHDO}) and M3 then imply A3. Thus, $\text{M$j$} \Rightarrow \text{A$j$}$ for all~$j$, proving part~(b).

The proof for part~(a) is by construction. We show that if $\{\op A(\bvec \alpha)\}$ is the valid PPO frame for the DWF in question, then the stencil ${M = \Wig^\twod[\op A(\bvec 0)] = \bar M}$ is valid and produces an $M$-DWF identical to the target one. Equation~\eqref{eq:P_on_Wig} verifies the second equality. The doubled DWF, limited to functions in the image of $\Wig^\twod$, is a faithful representation of the Hilbert-Schmidt inner-product space on operators~\cite{leonhardt_Discrete_1996}. Since $\{\op A(\bvec \alpha)\}$ is a valid PPO frame (by the statement of the Theorem), A1--A4 hold, so let us see what this implies in the doubled DWF representation. Since A4 holds, we have (see footnote~\ref{foot:A2d_WHDO})
\begin{align}
\label{eq:Mbar_is_Wig_A}
    \Wig^\twod[\op A(\bvec \alpha)] &= \bar M( \inputdot - 2\bvec \alpha), \\
    \text{and thus}\qquad\qquad
    \op A(\bvec \alpha) &= \Op^\twod[\bar M( \inputdot - 2\bvec \alpha)].
\label{eq:target_A_from_Mbar}
\end{align}
Furthermore, A1 implies M1, A2 implies M2, and A3 and A4 together imply M3, all using the features of the doubled DWF representation~\cite{leonhardt_Discrete_1996}. (We use the same properties as those discussed in the proof of part~(b), above, just with implication arrows reversed.) This proves that $\bar M$ is a valid stencil. Comparing Eqs.~\eqref{eq:AM_from_Mbar} and~\eqref{eq:target_A_from_Mbar}, along with M1, shows that the PPO frame generated by~$\bar M$ is identical to the target one. This proves part~(a) and, hence, the theorem.  \qed

\subsection{Proof of Corollary~\ref{corr:nonunique_stencil}}
The discussion around Eq.~\eqref{eq:AM_from_Mbar} shows that for all stencils~$M$, both~$M$ and its projection~$\bar M$ generate the same PPO frame. Furthermore, if two PPO frames are identical, they have the same PPO at $\bvec \alpha = \bvec 0$. By Eq.~\eqref{eq:Mbar_is_Wig_A}, this means the two projected stencils~$\bar M = \Wig^\twod[\op A(\bvec 0)]$ are identical. 
\qed

\subsection{Reduction and coarse-grain stencils: validity and PPO marginalisation}
For odd~$d$, validity of $M^\RS$ and marginalisation of $\op A^\RS$ follow from their equivalence with Gross' DWF~\cite{gross_Hudsons_2006}, Eq.~\eqref{PPOs_RS_Leon_Cohen_Woot1}.

For even~$d$, to prove validity of~$M^\CGS$, we need~$\bar M^\CGS$, obtained via Eq.~\eqref{eq:proj_explicit}. Since $d$ is even, $(-1)^{n d} = 1$ for any integer~$n$, and with the modular Kronecker enforcing $\bvec m = \bvec b' + d\, \bvec b \pmod {2d}$, we can replace~$m_j \mapsto b'_j$ in the exponent of~$(-1)$, giving
\begin{align}
    \bar M^\CGS_{\text{even $d$}}(\bvec m) &= \frac 1 4 \sum_{\bvec b, \bvec b'} (-1)^{b_1 b'_2 - b_2 b'_1} \Delta_{2d}[\bvec m - \bvec b' - d\, \bvec b],
\end{align}
which is real (M1) and sums to~1 (M2). Now we just need to check M3. The support of $\bar M^\CGS_{\text{even $d$}}$ is dictated entirely by the modular Kronecker condition. Displacing ${\bvec m \mapsto \bvec m + 2 \bvec \alpha}$ gives a function with non-overlapping support unless $2 \bvec \alpha = d\, \bvec b''$. Limiting to this case,
\begin{align}
    &(\bar M^\CGS_{\text{even $d$}} \star \bar M^\CGS_{\text{even $d$}})(d\, \bvec b'') = \frac 1 4 \sum_{\bvec b'} (-1)^{b_1'' b'_2 - b_2'' b'_1} = \Delta_2[\bvec b''],
\end{align}
which is true for all even $d$, verifying M3. Thus, $M^\CGS$ is valid for all even~$d$.

Next, we show that marginalisation holds for $\op A^\CGS$ in even dimensions. Using Eq.~\eqref{eq:AM_def}, we know that~${\op A^\CGS(\bvec \alpha) = \frac 1 2 \sum_{
\bvec b} A^\twod(2 \bvec \alpha + \bvec b)}$, which means that $\op Q^\CGS_V(\alpha_1)$ is just the sum of two neighbouring vertical marginals of the doubled PPO frame:
\begin{align}
    \op Q^\CGS_V(\alpha_1)
    &=
    \sum_{b_1} \op Q^\twod_V(2 \alpha_1 + b_1) 
    = 
    \qoutprod {\alpha_1} {\alpha_1},
\end{align}
where $\op Q^\twod_V(m_1) \coloneqq \frac 1 {2d} \sum_{m_2} \op A^\twod(\bvec m)$, and the second equality comes from the doubled DWF marginalisation property~\cite{leonhardt_Discrete_1996}, specifically $\op Q^\twod_V(m_1) = \Delta_2[m_1] \qoutprod {\frac {m_1} 2} {\frac {m_1} 2}$. Analogous results hold for horizontal PPO marginals and the $p$ basis (eigenstates of~$\op X$) by swapping subscripts on variables (${1 \leftrightarrow 2}$), operators (${V \leftrightarrow H}$), and kets (${q \leftrightarrow p}$). This proves the marginalisation condition for the CGS~PPO.

\subsection{Dirichlet Kernel stencil}
We define, for odd~$d$ only, the $2$-dimensional Dirichlet kernel stencil (DKS), inspired by Ref.~\cite{agam_Semiclassical_1995}:
\begin{align}
    M^\DKS(\bvec{m}) 
&\coloneqq 
    K_d(m_1) K_d(m_2), \\*
\text{where}\qquad
    K_d(m) &\coloneqq \frac
        {\sin(\frac {\pi}{2} m)}
        {d \sin(\frac {\pi}{2d} m)} 
        = 
        d^{-1}
        \sum_{\mathclap{t=-\frac 1 2(d-1)}}^{\mathclap{\frac 1 2 (d-1)}} \omega_{2d}^{tm}.
\label{eq:DK1d}
\end{align}
The summation range for~$t$ corresponds (for odd~$d$) to the condition $\Re \omega_{2d}^t > 0$, so we can define the $2d$-periodic boxcar function, $\Pi:\integers \to \integers_2$,
\begin{align}
    \Pi(k) &\coloneqq 
    H(\Re \omega_{2d}^k)
    =
    \begin{cases}
        1, & \text{if $\Re \omega_{2d}^k > 0$} \\
        0, & \text{otherwise},
    \end{cases}
\end{align}
where $H$ is the Heaviside step function. This reveals $K_d$ as $\sqrt{\frac 2 d}$ times the $2d$ discrete Fourier transform of $\Pi$. Since $\Pi$ is symmetric, the SDFT [Eq.~\eqref{eq:SDFT}] of $M^\DKS$ is then
\begin{align}
    \sft M^\DKS(\bvec m)
    &=
    \frac 2 d
    \Pi(m_1) \Pi(m_2)
    \eqqcolon
    \frac 2 d
    \Pi(\bvec m).
\end{align}
The simplicity of~$\sft M^\DKS$ allows us to illustrate the use of 
the stencil criteria in the Fourier domain, $\sft{\text M}1$--$\sft{\text M}3$. For this, we need the projection of $\sft M^\DKS$,
\begin{align}
    \bar{\sft M}^\DKS
    (\bvec m)
    &=
    \frac 1 4 \sum_{\bvec b} 
    (-1)^{b_1 m_2 -b_2 m_1 - b_1 b_2 d} \sft M^\DKS(\bvec m - d\, \bvec b)
    \nonumber \\
    &=
    \frac 1 {2d} \sum_{\bvec b} 
    (-1)^{b_1 m_2 -b_2 m_1 - b_1 b_2 d} \Pi(\bvec m - d \bvec b) \nonumber \\*
    &=
    \frac 1 {2d} 
    (-1)^{b_1 m_2 -b_2 m_1 - b_1 b_2}
    \Bigr\rvert_{b_j = 1 - \Pi(m_j)}.
\end{align}
Notice $\bar{\sft M}^\DKS(\bvec m) \in \mathbb{R}$ and $\bigl\lvert \bar{\sft M}^\DKS(\bvec m) \bigr\rvert = \bar{\sft M}^\DKS (\bvec 0) = \frac 1 {2d}$ for all~$\bvec m \in \mathcal P_{2d}$, verifying $\sft{\text M}1$--$\sft{\text M}3$. Thus, $M^\DKS$ is valid for odd~$d$. We leave to future work a discussion of its additional properties and comparison with other stencils.

\vfill

\bibliography{DWF_Paper_1}

\end{document}